\documentclass[runningheads]{llncs}

\usepackage{amsmath,amsfonts}
\usepackage{algorithm}
\usepackage[noend]{algpseudocode}
\usepackage{graphicx}
\usepackage{xspace}
\usepackage{tabularx}
\usepackage{booktabs}
\usepackage{textcomp}
\usepackage{todonotes}
\usepackage{nicefrac}
\usepackage[]{xcolor}
\usepackage{cancel}
\usepackage{cleveref}
\usepackage{rotating}
\usepackage{tikz}
\usepackage{csquotes}
\usepackage{booktabs,tabularx}
\usepackage{pifont}
\usepackage{pgfplots}
\usepackage{pifont}
\usepackage{caption}
\usepackage{subcaption}
\usepackage{microtype}
\pgfplotsset{compat=1.16}
\usepgfplotslibrary{statistics}
\usetikzlibrary{positioning,arrows.meta,arrows,shapes.misc,intersections,fit,decorations.text}
\usetikzlibrary{decorations.pathreplacing,calligraphy,spy}
\lccode`/`/
\newenvironment{proofsketch}{%
  \renewcommand{\proofname}{Argument}\proof}{\endproof}

\newcommand{\furl}[1]{\footnote{\url{#1} --- Accessed \today}}

\def\faa{\mathcal{F}\xspace}

\def\prot{\textsc{Panini}\xspace}

\newcommand{\cauth}[1]{
    \textit{Ch}_{\text{auth}}\ifx\yyy#1\yyy\else\left(#1\right)\fi
}
\newcommand{\canon}[1]{
    \textit{Ch}_{\text{anon}}\ifx\yyy#1\yyy\else\left(#1\right)\fi
}
\def\new{\color{black}}

\algnewcommand\algorithmicon{\textbf{on}}
\algnewcommand\algorithmicfrom{\textbf{from}}
\algblockdefx[ON]{On}{EndOn}[2]
  {\algorithmicon\ #1\ \algorithmicfrom\ #2\ \algorithmicdo}

\makeatletter
\ifthenelse{\equal{\ALG@noend}{t}}%
  {\algtext*{EndOn}}
  {}%
\makeatother

\begin{document}

\title{\prot\ --- Anonymous Anycast and an Instantiation
}
\subtitle{Extended Version}
\author{
    Christoph Coijanovic\inst{1} \and 
    Christiane Kuhn\inst{2} \and
    Thorsten Strufe\inst{1}
}

\authorrunning{Coijanovic et al.}

\institute{
    Karlsruhe Institute of Technology \email{firstname.lastname@kit.edu} \and
    NEC Laboratories Europe \email{firstname.lastname@neclab.eu}
}

\maketitle

\begin{abstract}
Anycast messaging (i.e., sending a message to an unspecified receiver) has long been neglected by the anonymous communication community.
An \emph{anonymous} anycast prevents senders from learning who the receiver of their message is,  allowing for greater privacy in areas such as political activism and whistleblowing.
While there have been some protocol ideas proposed, formal treatment of the problem is absent.
Formal definitions of what constitutes anonymous anycast and privacy in this context are however a requirement for constructing protocols with provable guarantees.
In this work, we define the anycast functionality and use a game-based approach to formalize its privacy and security goals.
We further propose \prot, the first anonymous anycast protocol that only requires readily available infrastructure.
We show that \prot allows the actual receiver of the anycast message to remain anonymous, even in the presence of an honest but curious sender. 
In an empirical evaluation, we find that \prot adds only minimal overhead over regular unicast: 
Sending a message anonymously to one of eight possible receivers results in an end-to-end latency of 0.76s.
\end{abstract}

\textit{
    This is the extended version of ``Panini --- Anonymous Anycast and an Instantiation'' published at ESORICS 2023.
    Compared to the standard version, it additionally contains a more in-depth introduction to provable privacy (\Cref{sec:bg:provable_privacy}) and symmetric encryption (\Cref{sec:bg:enc}), proofs of the relations between our defined privacy notions (\Cref{sec:prob:imp}), a pseudocode description of \prot (\Cref{alg:prot}), proof that \prot achieves message confidentiality (\Cref{thm:atu:mc}), and long-term latency measurements of Nym (\Cref{sec:eval:nym}).  
}
\section{Introduction}
\label{sec:intro}
In an \emph{anycast}, messages are received by any one of a group of eligible receivers.
This communication pattern is widely used in domain name resolution and content delivery networks~\cite{RefAnycastStats}.
Because the actual receivers are not predetermined, anycast also lends itself naturally to anonymous communication:
Consider a group of political activists who fear retribution from the opposing regime.
The activists want to implement \emph{dead man's switches} among themselves, i.e., if anyone is caught, someone else is notified and can take over their duties.
If the arrested activist played a prominent role in the opposition, the regime will be particularly interested in her replacement.
We can derive two main requirements from sending the dead man's notification via anonymous anycast:
\begin{enumerate}
    \item No one (including the anycast sender) should be able to identify the receiver.
        This way, the captured activist cannot be forced to reveal her successor.
        Note that to hide this information, the receiver must be chosen non-deterministically.
    \item The set of possible receivers should be constrainable by the sender.
        This ensures that one of the activist's trusted allies becomes her successor.
\end{enumerate}
A third non-functional requirement is that sending an anonymous anycast should be as easy to set up as possible.
Any obstacles, such as the need to set up a server infrastructure, will limit adoption.

Of course, anonymous anycast is not limited to political activism.
Anonymous anycast is preferable to the much more common anonymous \emph{unicast}~\cite{RefVuvuzela,RefExpress,RefTor} in any setting where receiver information should be hidden from the sender:
In whistleblowing, anonymous anycast provides plausible deniability for the sender.
In online lotteries, distributing the winning token via anonymous anycast guarantees that the winner is chosen without bias.
In distributed computing, anonymous anycast protects the receiver from targeted denial of service attacks.

While not receiving the same amount of attention as anonymous unicast and multicast, there is some academic literature that addresses similar issues.
Mislove et al. mention that their AP3 protocol can be extended to support anycasts~\cite{RefAP3}.
A recent line of research~\cite{RefBBG,RefETF,RefRPIR} focuses on anonymously selecting committee members to receive messages. 
We see one major shortcoming in the related work:

Related work considers anonymous anycast from a protocol perspective.
To our knowledge, no work has focused on the formal aspects of the problem.
Without a formal understanding of the properties of anonymous anycast, it is difficult to compare current and future protocols when each defines its own `flavor' of anonymous anycast.
Thus, our goal is to provide a concrete definition of the anonymous anycast problem and to formally define the main privacy goals of an anonymous anycast system.

We identify Message Confidentiality, Fairness, and Receiver Anonymity as the main goals of anonymous anycast.
We formalize these goals using a game-based approach that is common in cryptography (e.g., with the IND-CPA notion from semantic security~\cite{RefIndCPA}) and already well established in anonymous unicast communication~\cite{RefNotions}.
Our game-based privacy goals are unambiguously defined and allow for rigorous analysis of anycast protocols.

In this paper, we also propose \prot, an anonymous anycast protocol, to show that our defined privacy goals are achievable by efficient protocols.
\prot relies on a readily available infrastructure:
An authenticated unicast channel (e.g., Signal\furl{https://signal.org}) and a unicast channel that unlinks senders from their messages (e.g., Nym\furl{https://nymtech.net}).
The unlinking unicast channel is used by possible receivers to send randomness to the anycast sender.
Based on this randomness, the sender can choose a receiver without learning its identity.

In summary, the main contributions of this paper are
\begin{itemize}
    \item The formalization of functionality and privacy goals in anonymous anycast
    \item The proposal of \prot, the first protocol that allows anonymous anycast over readily available infrastructure
    \item A security analysis of \prot, showing that it achieves our previously defined privacy goals.
    \item In-depth empirical evaluation of \prot including long-term latency measurements of Nym.
\end{itemize}

The rest of this paper is organized as follows:
\Cref{sec:related} introduces the related work in greater detail.
\Cref{sec:bg} presents the necessary background on provable privacy, linkable ring signatures, and Nym.
\Cref{sec:prob} contains our formal treatment of the anonymous anycast problem including definitions of the privacy goals.
\Cref{sec:prot} describes the \prot protocol.
\Cref{sec:eval} contains \prot's empirical evaluation.
Finally, \Cref{sec:conc} concludes the paper.

\section{Related Work}
\label{sec:related}
Recall our requirements for anonymous anycast:
(1) no entity (including the anycast sender) should learn who is receiving the anycast message, (2) the set of possible receivers should be restrictable, and (3) the anycast should be as easy to set up as possible.
In this section, we will present anonymous anycast-related work and discuss whether it meets these requirements.

\paragraph{Target-Anonymous Channels}
A recent line of work considers \emph{target-anonymous channels}~\cite{RefBBG,RefRPIR}.
Benhamouda et al. informally define a target-anonymous channel as one that allows ``anyone to post a message to an unknown receiver''~\cite{RefBBG}.
Both papers construct the target-anonymous channel using the same basic technique:
One protocol participant is chosen to select the receiver of the channel.
The participant then provides all other participants with a way to contact the receiver without revealing the receiver's identity to them.
Since the selecting participant inherently learns who the receiver will be in future uses of this target-anonymous channel, our first anonymous anycast requirement is not satisfied.

\paragraph{AP3}
AP3~\cite{RefAP3} is a mix network that implements the publish/subscribe communication pattern.
Publishers and subscribers are both connected through the mix network to a common root node.
The root node receives messages from the publishers and forwards them to the subscribers.
Mislove et al. do not discuss in detail how AP3 can be extended to provide anycast functionality.
We assume, based on the available information, that the root node randomly selects a subset of subscribers as actual receivers, rather than forwarding to all. 

In AP3, the anycast sender must trust the root node to perform the anycast correctly (e.g., not send the message to all users).
AP3's authors do not mention that the ability to subscribe to a publisher is limited.
Thus, there seems to be no way for the anycast sender to define the set of possible receivers.
So our second requirement is not satisfied.

\paragraph{Encryption to the Future}
Encryption to the Future (EtF)~\cite{RefETF,RefETF2,RefETF3} is a cryptographic primitive where messages can be encrypted for a given \emph{role}, rather than for a specific receiver.
Later, a lottery is held to determine who gets to hold the role and thus be able to decrypt the ciphertext.
To the best of our knowledge, suitable lottery primitives are all based on proof-of-stake blockchains~\cite{RefAnonLottery1,RefAnonLottery2}.

Even assuming the general availability of a suitable blockchain, the requirement that users acquire a cryptocurrency stake in order to receive an anycast message is a significant barrier to adoption.
Thus, our third requirement is not met.

\section{Background}
\label{sec:bg}
Next, we introduce the necessary background for the remainder of this paper.
This section is divided into background on provable privacy (\Cref{sec:bg:provable_privacy}), linkable ring signatures (\Cref{sec:bg:ring_signature}), {\new symmetric encryption (\Cref{sec:bg:enc}),} and the Nym anonymous communication protocol (\Cref{sec:bg:nym}).

\subsection{Provable Privacy}
\label{sec:bg:provable_privacy}
When designing protocols that handle sensitive information, privacy and security must be carefully considered.
Rather than designing protocol mechanisms that seem reasonable and hoping for the best, it is desirable to have concrete proof that the protocol protects sensitive information.
To provide such proof, however, it is necessary to unambiguously define what information is to be protected.
In cryptography, such formal definitions of security have been established since the 1980s~\cite{RefProvableSecurity}.
Message confidentiality is formalized using indistinguishability games such as IND-CPA, IND-CCA1, and IND-CCA2~\cite{RefIndCPA}.
In these games, an adversary can choose two messages to send to a challenger.
The challenger chooses one of these messages at random and encrypts it.
The adversary must determine from the resulting ciphertext which of his messages was encrypted.
If it can be shown that there is no adversary who can determine the correct message with a non-negligible advantage over random guessing, then it has been proved that the ciphertexts of the encryption scheme do not reveal any information about the contained plaintexts.

In privacy, there is a much wider variety of possible goals than in classical security:
Some protocols may focus on protecting the privacy of the sender, while others may consider the receiver, or both.
Indeed, many provable privacy frameworks have been proposed~\cite{RefNotions,RefProvablePrivacy1,RefProvablePrivacy2,RefProvablePrivacy3}.
We base our formalized anycast privacy on the work of Kuhn et al.~\cite{RefNotions}, as their framework can express all the previous goals and bases them on a common indistinguishability game.

Similar to the two messages in the IND-CPA game, the adversary in the Kuhn et al. game may choose two scenarios, each consisting of a series of communications.
Each communication is defined by its sender, receiver, and message.
The adversary chooses a scenario at random and simulates the protocol execution of the enclosed communications.
The adversary receives any protocol output from the challenger that her abilities would allow her to learn from a real protocol run:
For example, adversaries that can \emph{globally} observe will receive data from every network link, while local observers will only receive data from links in their domain.
Based on this output, the adversary must decide which of her scenarios has been executed.

Different privacy goals are expressed by restrictions on how communication may differ between scenarios:
For example, consider a protocol that aims to protect only the activity of the sender.
If the adversary can distinguish between the scenarios based on information about the active receivers, she has an unfair advantage and can win the game, even though the protocol achieves its intended goal.
Thus, to accurately model the protocol's goal, the adversary is restricted to submitting scenarios where the communications \emph{only} differ in their senders, i.e., receivers and messages must be identical between scenarios.
If the adversary can still distinguish between scenarios, there must be some disclosure of sender activity, implying that the protocol does not meet its goal.

A common goal of anonymous communication protocols is to unlink senders from their messages~\cite{RefTrellis,RefSabre,RefExpress,RefOrgAn}.
Kuhn et al. formalize this goal in the privacy notion of Sender-Message Pair Unlinkability $(SM)\overline{L}$.
Intuitively, a protocol that achieves $(SM)\overline{L}$ can reveal which senders are active and even which messages are being sent, but not who is sending which message.
Consider the example scenarios presented in \Cref{tab:sml}:
In each scenario, the same senders ($A$ and $C$) are active, and the same messages ($m_1$ and $m_2$) are sent to the same receiver ($B$).
The only difference between the scenarios is who sends which message.
For each scenario, an alternative instance is introduced that reverses the order of communication.
Instances ensure that the adversary cannot distinguish between scenarios based on who sends first.
The adversary randomly selects an instance of a scenario and simulates the protocol with the communications it contains.
Note that the adversary only has to determine the selected scenario, not the instance.

\begin{table}[]
    \centering
    \begin{tabularx}{0.65\columnwidth}{lXX}\toprule
        & \textbf{Scenario 0} & \textbf{Scenario 1}\\\midrule
         \textbf{Instance 0} & \begin{tabular}{c}$A\xrightarrow{m_1} B$\\$C\xrightarrow{m_2} B$\end{tabular} & \begin{tabular}{c}$A\xrightarrow{m_2}B$\\$C\xrightarrow{m_1}B$\end{tabular}\\
         \textbf{Instance 1} & \begin{tabular}{c}$C\xrightarrow{m_2} B$\\$A\xrightarrow{m_1} B$\end{tabular} & \begin{tabular}{c}$C\xrightarrow{m_1}B$\\$A\xrightarrow{m_2}B$\end{tabular}\\\bottomrule
    \end{tabularx}
    \caption{
        Valid scenarios for Sender-Message Pair Unlinkability $(SM)\overline{L}$.
        $A\xrightarrow{m}B$ denotes that user $A$ sends message $m$ to user $B$.
    }
    \label{tab:sml}
\end{table}

\subsection{Linkable Ring Signatures}
\label{sec:bg:ring_signature}
\emph{Linkable ring signature} schemes{\new~\cite{RefFirstLRS}} are an extension of standard digital signature schemes:
Compared to standard digital signatures, ring signatures allow verification against a set of multiple verification keys.
Verification succeeds if the signature was made with one of the corresponding secret keys.
The linkability property introduces an additional algorithm that determines whether two signatures were created using the same secret key.

Like standard signature schemes, linkable ring signatures provide \emph{unforgeability}, i.e. valid signatures can only be made by users who have a secret key that matches a verification key in the ring.
Like all ring signature schemes, linkable ring signatures also provide \emph{signer anonymity}, i.e. the identity of the actual signer cannot be determined, with a non-negligible advantage over random guessing from the signature or verification process (as long as the signer's secret key is not known to the adversary).
Finally, \emph{linkability} guarantees that the linking algorithm does not return false-positive or false-negative results.
For a more formal definition of these properties, see Liu et al.'s model of a linkable ring signature system~\cite{RefLRSDef}.
Specifically, a linkable ring signature scheme consists of the following algorithms~\cite{RefRingSig2}:
\begin{itemize}
    \item \textsc{Sig.Setup}$(1^\lambda) \to pp$: 
        On input of security parameter $1^\lambda$, return public parameters $pp$.
    \item \textsc{Sig.KeyGen}$(pp) \to (vk,sk)$:
        On input of public parameters $pp$, returns a pair of public and secret key $(vk, sk)$.
    \item \textsc{Sig.Sign}$(sk,m,R) \to \sigma$:
        On input of a secret key $sk$, a message $m$, and a ring $R = \{vk_0,\dots,vk_\ell\}$, output a signature $\sigma$.
    \item \textsc{Sig.Verify}$(\sigma,m,R) \to \{0,1\}$:
        On input of a signature $\sigma$, a message $m$, and a ring $R = \{vk_0,\dots,vk_\ell\}$, output:
            \begin{itemize}
                \item 1 (accept), iff $\sigma$ was generated by executing \textsc{Sign}$(sk,m,R)$, where $sk$ corresponds to some public key in $R$.
                \item 0 (reject), else.
            \end{itemize}
    \item \textsc{Sig.Link}$(\sigma,\sigma^\prime) \to \{0,1\}$:
        On input of two signatures $\sigma$ and $\sigma^\prime$ output 1, iff $\sigma$ and $\sigma^\prime$ were created using the same secret key and 0 otherwise.
\end{itemize}

{\new%
\subsection{Symmetric Encryption}
\label{sec:bg:enc}
Symmetric ciphers such as AES~\cite{RefAES} are ubiquitous in modern communication because they provide confidentiality with low overhead.
A symmetric cipher consists of the following algorithms:
\begin{itemize}
    \item \textsc{Cipher.KeyGen}\((1^\lambda) \to k\):
        On input of security parameter \(1^\lambda\), return a key \(k\).
    \item \textsc{Cipher.Enc}\((m,k) \to c\):
        On input of plaintext \(m\) and key \(k\), return ciphertext \(c\).
    \item \textsc{Cipher.Dec}\((c,k) \to m\):
        On input of ciphertext \(c\) and key \(k\), return plaintext \(m\).
\end{itemize}
}

\subsection{Nym}
\label{sec:bg:nym}
Most proposed anonymous communication networks exist only on paper, and do not provide a public instance for people to use.
A recent exception is Nym~\cite{RefNym}, which provides both client software to download and servers to connect to\furl{https://nymtech.net}.

{\new While Nym has not yet been subjected to much scientific scrutiny, it adopts the communication architecture of the well-established Loopix mix network design~\cite{RefLoopix}.}
In mix networks, messages are wrapped in multiple layers of encryption and sent through a series of mix nodes.
Each node removes the outermost layer of encryption before passing it on to the next node in the path.
After the last layer of encryption is removed, the message is forwarded to the intended receiver.
In addition to decryption, mix nodes also delay the forwarding of messages for a random amount of time.
Nym also uses cover traffic from both clients and servers to 1) hide communication patterns, 2) detect denial of service attacks, and 3) ensure that mix nodes have a sufficient amount of alternative traffic in which to hide messages.
The goal of a mix network is to unlink messages from their senders.
The layered encryption ensures that they cannot be linked based on message content, while the randomized delays ensure that they cannot be linked based on timing (given a sufficient amount of alternative messages traversing the mix node).
As long as there is at least one honest mix node between sender and receiver, this goal (and thus $(SM)\overline{L}$) is achieved.

{\new%
Nym differentiates itself from Loopix by introducing a blockchain-based system that allows users to anonymously pay for access to the system, and rewards nodes for mixing.
The payment system is optional and not yet deployed in the production version.
While this addition of payments may provide real-world benefits, it is independent of the actual communication infrastructure, is not relevant to our use of the system, and is therefore not considered in the remainder of this paper.`
}

Compared to Tor~\cite{RefTor}, the best-known anonymous communication network, Nym's use of randomized delays combined with additional cover traffic ensures unlinkability even against global observers.
Tor's vulnerability to traffic analysis by global observers is well documented~\cite{RefTorSurvey}, making Nym preferable in environments where such an adversary may be present.

\section{Problem Definition}
\label{sec:prob}
In this section, we consider the anonymous anycast problem from a formal perspective.
\Cref{sec:prob:func} defines which functionality an anycast protocol has to provide to be considered \emph{correct}.
\Cref{sec:prob:adv} introduces our assumed adversary model.
\Cref{sec:prob:goals} contains game-based formalizations of anycast privacy goals.

In the following, we use $X \subset_i Y$ to express that set $X$ is a strict subset of set $Y$ consisting of $i$ elements.
$X\subseteq_i Y$ is used analogously.
Further, $|X|$ expresses the number of elements in $X$.
We use $U$ to denote the set of all protocol participants.

\subsection{Functionality}
\label{sec:prob:func}
An anonymous anycast is a protocol between $q = |U|$ users.
The anycast's \emph{sender} selects a set of $l\leq q$ \emph{possible receivers} $U_p$, of which $n\leq l$ shall receive the anycast message.
The anycast functionality then selects $n$ \emph{actual receivers} $U_a$ out of the set of possible receivers at random and forwards the message to them.
This functionality can be trivially provided by a trusted third party $\faa$.
\Cref{def:if} describes $\faa$'s behavior and bases anycast correctness on equivalence to it. 

\begin{definition}[Anonymous Anycast Correctness]
    \label{def:if}
    The anonymous anycast functionality $\faa$ interacts with a set of $q$ users $U = \{u_1$, $\dots$, $u_q\}$.
    It behaves as follows:
    \begin{enumerate}
        \item $\faa$ waits for input of form $(m,n,U_p)$ from sender $u_s\in U$.
            $m$ denotes the message, $U_p \subseteq_l U$ the set of \emph{possible} receivers and $n\in\{1,\dots,l\}$ the requested number of actual receivers.
        \item $\faa$ selects $U_a\subset_n U_p$ uniformly at random.
        \item $\forall u\in U_a$: $\faa$ sends $m$ to $u$.
    \end{enumerate}
    An anonymous anycast protocol is \emph{correct} if it provides the same functionality as $\faa$.
\end{definition}

Following \Cref{def:if}, we consider a multicast (i.e., $n=l$/all possible receivers are selected as actual ones) a special form of anycast.

\subsection{Adversary Assumptions}
\label{sec:prob:adv}
We assume an adversary $\mathcal{A}$ who can globally observe any network link, as well as actively interfere (i.e., drop, delay, modify, insert, and replay) with arbitrary packets. 
We further assume that $\mathcal{A}$ can corrupt the sender as well as a fraction of possible receivers.
We assume that corrupted users are honest but curious. 
We exclude arbitrarily malicious users, since they can trivially bypass any protocol's protection mechanism and send the anycast message directly to a receiver of their choice.

\subsection{Privacy Goals}
\label{sec:prob:goals}
In general, we must assume that both senders and receivers of an anycast are of interest to an adversary.
However, to express sender-related privacy goals, existing notions of privacy for unicast communication can be used~\cite{RefNotions}.
Thus, when considering the privacy goals of anonymous anycast, we focus on the receiver side.
We propose the following three main goals for our anycast setting:
\begin{enumerate}
    \item \textbf{Message Confidentiality ($MC$).}
        Outside of the sender and actual receivers, nobody shall learn information about the content of the anycast message.
    \item \textbf{Receiver Anonymity ($RA$).}
        Any adversary shall only learn trivial information about the \emph{actual} receivers. 
        Trivial information includes, for example, that a user learns that she is an actual receiver.
    \item \textbf{Fairness $(F)$.}
        Any possible receiver shall be chosen as the actual receiver with the same probability, except for negligible deviations.
        In a protocol without fairness, an adversary learns that some users are more likely to receive the anycast message, even without observation.
\end{enumerate}
While these informal descriptions of our privacy goals give a good intuition of the information that should be protected, stating them informally is not sufficient to prove that a protocol achieves them.
{\new Thus, we use a game-based approach to formalize our privacy goals next.}

{\new%
Our games have a common structure:
They are played between a challenger \(\mathcal{C}\) and an adversary \(\mathcal{A}\).
The challenger internally simulates the anycast protocol \(\Pi\).
\(\mathcal{A}\) can provide input (a \enquote{challenge}) to the protocol and receives its output from \(\mathcal{C}\).
Based on the output, \(\mathcal{A}\) must determine some information about the protocol execution.
}

\paragraph{Formalizing Message Confidentiality}
To formalize message confidentiality, we build on Kuhn et al.'s privacy game (see \Cref{sec:bg:provable_privacy}).
We need to make the following modifications to adapt to the anycast setting:
\begin{enumerate}
    \item While unicast communications are defined by a sender, a message, and a single receiver, an anycast message is defined by the sender, the message, the number of intended receivers, and the set of possible receivers.
    Thus, communications are expressed as tuples $(s,m,n,U_p)$.
    \item We assume that $\mathcal{A}$ is able to corrupt both the sender and a fraction of the receivers.
    Note that the anycast sender trivially learns the content of the message.
    Possible receivers also learn the message content if they are selected as actual receivers.
    To give the protocol a fair chance of achieving message confidentiality, $\mathcal{C}$ only provides protocol output to $\mathcal{A}$ if the sender and the actual receivers are not corrupted.
\end{enumerate}
The resulting game $\mathcal{G}_{MC}$ for anycast protocol $\Pi$ proceeds as follows:
\begin{enumerate}
    \item $\mathcal{C}$ selects a challenge bit $b\in\{0,1\}$ uniformly at random
    \item $\mathcal{A}$ submits a challenge $Ch = (s,\{m_0,m_1\},n,U_p)$
    \item $\mathcal{C}$ simulates the anycast protocol $\Pi$'s execution of ($s$, $m_b$, $n$, $U_p$) and saves the set of actual receivers $U_a$ as well as $\Pi$'s output $\Pi$($s$, $m_b$, $n$, $U_p$)
    \item If $\forall u \in U_a \cup \{s\} : u$ is \emph{not} corrupted, $\mathcal{C}$ forwards $\Pi$($s$, $m_b$, $n$, $U_p$) to $\mathcal{A}$
    \item $\mathcal{A}$ submits her guess $b^\prime \in\{0,1\}$ to $\mathcal{C}$.
        $\mathcal{A}$ wins if $b = b^\prime$ and looses otherwise.
\end{enumerate}

Analogous to Kuhn et al.'s game, steps (2-4) can be repeated an arbitrary number of times (with different challenges) to allow $\mathcal{A}$ to adapt its strategy based on its observations.

We say that a protocol $\Pi$ achieves message confidentiality if there is no probabilistic polynomial-time algorithm $\mathcal{A}$ that can win $\mathcal{G}_{MC}$ with a non-negligible advantage over random guessing.
Since there are two possible values for $b$, random guessing has a probability of success of \nicefrac{1}{2}.

\paragraph{Receiver Anonymity}
At first glance, one could define a receiver anonymity game analogous to the message confidentiality game:
$\mathcal{A}$ submits two possible sets of actual receivers and must decide which one was chosen by the challenger.
However, since actual receivers are supposed to be chosen non-deterministically, they cannot be set by the challenger.
We adapt the game model so that the adversary has to make a guess for an actual receiver instead of making a binary decision.

Note that, as with message confidentiality, user corruption may allow $\mathcal{A}$ to trivially determine an actual receiver.
If one of the users corrupted by $\mathcal{A}$ receives the anycast message, $\mathcal{A}$ unambiguously learns that this user was chosen as the actual receiver.
If $\mathcal{A}$ has corrupted all but $n$ users, and none of them receives the anycast message, then all remaining users must be actual receivers.
Thus, the challenger must check if one of these conditions is true after the actual receivers have been selected, and stop the game accordingly.

The complete $\mathcal{G}_{RA}$ game proceeds as follows:
\begin{enumerate}
    \item $\mathcal{A}$ submits a challenge $Ch = (s,m,n,U_p)$
    \item $\mathcal{C}$ simulates the protocol $\Pi$'s execution of $Ch$ and saves the chosen actual receivers $U_a$.
        $\mathcal{C}$ checks if $\mathcal{A}$ can trivially win due to user corruption.
        This is the case if there is a corrupted $u\in U_a$, or if \emph{all} $u\in U_p\cap U_a$ are corrupted.
        In case of a trivial win, $\mathcal{C}$ discards $\Pi$'s output $\Pi(s,m,n,U_p)$.
        Otherwise, $\Pi(s,m,n,U_p)$ is forwarded to $\mathcal{A}$.
    \item $\mathcal{A}$ can either choose to (a) unveil the challenge or (b) submit her guess $u^\ast \in U$
        \begin{enumerate}
            \item If $\mathcal{A}$ requested to unveil the challenge, $\mathcal{C}$ forwards $U_a$ to $\mathcal{A}$
            \item If $\mathcal{A}$ submitted her guess, $\mathcal{C}$ checks if $u^\ast \in U_a$.
                If so, $\mathcal{A}$ wins and loses otherwise.
        \end{enumerate}
\end{enumerate}

To allow $\mathcal{A}$ to adapt her strategy, steps (1-3) can be repeated a polynomial number of times as long as $\mathcal{A}$ chooses to unveil the challenge.
Once $\mathcal{A}$ has submitted her guess, the game ends.
See \Cref{fig:gra} for a visualization of $\mathcal{G}_{RA}$.

\begin{figure}
    \centering
    \begin{tikzpicture}[
    msg/.style={midway,fill=white,font=\tiny,inner sep=1},
    cmp/.style={draw,fill=white,align=left,font=\tiny,inner sep=2},
    scale=0.7
]
    \node[font=\large] (c) at (0,7) {$\mathcal{C}$};
    \node[font=\large] (a) at (6,7) {$\mathcal{A}$};
    \draw[-|] (c) -- (0,1.75);
    \draw[-|] (0,1.5) -- (0,0.5);
    \draw[] (a) -- (6,2.5);
    
    \draw[-|] (6,2.5) to[out=-135,in=90] (5.75,2.25) -- (5.75,1.75);
    \draw (6,2.5) to[out=-45,in=90] (6.25,2.25) -- (6.25,0.5);
    \node[fill=white,inner sep=1,draw,circle,font=\footnotesize] (a0) at (6,2.5) {\textbf{or}};
    
    \node[cmp,xshift=0.47cm] (c0) at (0,6.4) {$b \gets^R \{0,1\}$};
    \draw[->] (6,6) -- node[msg] {$\textit{Ch} = (s,m,n,U_p)$} ++(-6,0);
    
    \draw[fill=black!10] (2,5.75) rectangle node[midway,font=\large] {$\Pi$} (4,4.75);
    \draw[->] (0,5.5) -- node[msg] {\textit{Ch}} ++(2,0); 
    \draw[->] (2,5) -- node[msg] {$U_a, \Pi(\textit{Ch})$} ++(-2,0); 
    
    \node[cmp,xshift=1.17cm] (c1) at (0,4) {
        \textbf{if} $\exists\: \text{cor. } u \in U_a$\\
        \textbf{or} $\forall\: u \in U_p\cap U_a : u \text{ cor.}$:\\
        $\Pi(Ch) \gets \bot$
    };
    
    \draw[->] (0,3) -- node[msg] {$\Pi(\textit{Ch})$} ++(6,0);
    
    \draw[->] (5.75,2) -- ++(-5.75,0);
    \node[fill=white,font=\tiny,inner sep=1] (mid1) at (3,2) {$u\in U_p$};
    
    \draw[->] (6.25,1.25) -- ++(-6.25,0);
    \node[fill=white,font=\tiny,inner sep=1] (mid2) at (3,1.25) {\texttt{unveil}};
    
    \draw[->] (0,0.75) -- ++(6.25,0);
    \node[fill=white,font=\tiny,inner sep=1] (mid3) at (3,0.75) {$U_a$};
    
    \draw[->] (6.25,0.5) -- (6.5,0.5) -- (6.5,6.25) -- (6,6.25);
\end{tikzpicture}
    \caption{
        Game $\mathcal{G}_{RA}$ formalizing receiver anonymity.
    }
    \label{fig:gra}
\end{figure}
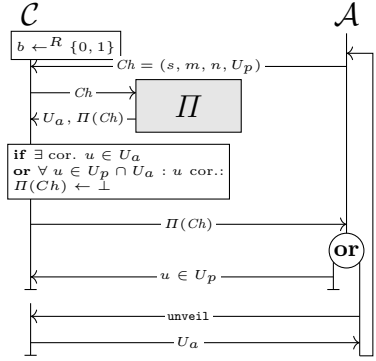

Analogous to message confidentiality, $\Pi$ achieves receiver anonymity if there is no probabilistic polynomial-time algorithm $\mathcal{A}$ that can win $\mathcal{G}_{RA}$ with a non-negligible advantage over random guessing.
Note that random guessing gives a probability of success of $\nicefrac{n}{|U_p|}$, not $\nicefrac{1}{2}$.

\paragraph{Formalizing Fairness}
Fairness is closely related to receiver anonymity:
If the protocol is not fair and favors some receivers over others, this information can be used by $\mathcal{A}$ to gain an advantage in $\mathcal{G}_{RA}$.
However, a protocol could be perfectly fair but fail to achieve receiver anonymity;
Consider a toy protocol that chooses the actual receivers uniformly at random, but then announces them publicly.

If the protocol is not fair, $\mathcal{A}$ should have an advantage in guessing the actual receivers without relying on the protocol output.
Thus, $\mathcal{G}_F$ differs from $\mathcal{G}_{RA}$ only in that $\mathcal{A}$ must submit her guess with the challenge, \emph{prior} to receiving the protocol output.
The resulting game is similar to the EUF-CMA game used to test the unforgeability of digital signatures~\cite{RefEUFCMA}.

{\new%
\paragraph{User Corruption}
As described in \Cref{sec:prob:adv}, the adversary has the ability to corrupt the anycast sender as well as a fraction of the possible receivers.
This ability is implemented in the games via a special challenge:
Instead of sending a challenge \textit{Ch}, \(\mathcal{A}\) can send a \emph{corruption query} specifying which users to corrupt in future runs.
In response, \(\mathcal{C}\) returns the internal state of the specified users.
The protocol output \(\Pi(\textit{Ch})\) may also change in future runs.
}

\subsection{Implications Between Notions}
\label{sec:prob:imp}
We have noted above that a protocol that achieves receiver anonymity has to be fair, but a fair protocol does not necessarily achieve receiver anonymity.
Thus, there is an interesting asymmetric relation between the two notions.
We say that there is an \emph{implication} between two notions if any protocol that achieves the implying notion also achieves the implied one.
Conversely, if there is no implication between two notions, there exists a protocol that achieves one but not the other.
Being aware of these implications is especially valuable when analyzing concrete protocols, as it reduces the number of notions that need to be considered.

To prove that notion $X$ implies notion $Y$, one assumes that there is a protocol that achieves $Y$, but not $X$.
Next, one has to show that any attack that an adversary could execute for $Y$ is also valid for $X$.
It follows then that this adversary can also break $X$, which contradicts the initial assumption.
Thus, there cannot be an adversary who can only break $Y$ but not $X$, so $X$ indeed implies $Y$.

\paragraph{Implications}
There is only one implication between our three anycast privacy notions:
Receiver anonymity implies fairness.
Intuitively, this makes sense as both are based on the same game and only differ in when they provide information to the adversary.

\begin{theorem}
    \label{thm:imp1}
    Receiver anonymity implies fairness.
\end{theorem}
\begin{proof}
    Assume a protocol $\Pi$ that achieves fairness, but not receiver anonymity.
    Recall that $\mathcal{G}_{F}$ and $\mathcal{G}_{RA}$ only differ in that the adversary has to submit her guess before receiving protocol output in $\mathcal{G}_F$.
    If there is an adversary $\mathcal{A}$ who can break fairness for $\mathcal{G}_F$, $\mathcal{A}$ can also break receiver anonymity by discarding the protocol output and behaving identically otherwise. 
    Thus any attack that breaks $\mathcal{G}_F$ also works in $\mathcal{G}_{RA}$ which contradicts the assumption.
\end{proof}

\paragraph{Non-Implications}
To ensure that all implications have been found, one also has to prove that there are \emph{no} implications between all remaining pairs of notions.
To construct a proof that $X$ does not imply $Y$, one proceeds as follows:
\begin{enumerate}
    \item Assume a protocol $\Pi$ that achieves notion $Y$
    \item Construct a protocol $\Pi^\prime$ that behaves identically to $\Pi$, except for the disclosure of some information that is protected by $X$ but not by $Y$.
    \item $\Pi^\prime$ still achieves $Y$, as the disclosed information cannot be used in $Y$'s game to gain an advantage by definition.
    \item $\Pi^\prime$ does not achieve $X$, as the disclosed information enables the adversary to win $X$'s game.
\end{enumerate}
As $\Pi^\prime$ is a protocol which achieves $Y$ but not $X$, $X$ cannot imply $Y$.

\begin{theorem}
    There are no relations among the anycast privacy notions except the one stated in \Cref{thm:imp1}.
\end{theorem}
\begin{proof}
    Refer to the following to see which proof applies to which relation.
    `$\to$' denotes that the row's notion implies the column's notion.
    (PX) denotes that the proof follows from transitivity and PX.
    \begin{center}
    \begin{tabularx}{0.6\columnwidth}{XXXX}\toprule
        & $MC$ & $F$ & $RA$ \\\midrule 
        $MC$  & = & P2 & P3 \\
        $F$   & P1 & = & P5 \\
        $RA$ & (P1) & $\to$ & = \\\bottomrule
    \end{tabularx}
    \end{center}
    For all non-implications between notions $X$ and $Y$, we need to proof that there exists a protocol that there exists a protocol which achieves $X$ but not $Y$.
    
    \textbf{P1.}
    Let $\Pi$ be an anycast protocol that achieves $F$.
    Let $\Pi^\prime$ be a protocol that behaves identical to $\Pi$, but publishes the message content after every anycast.
    $\Pi^\prime$ still achieves $F$ as $\mathcal{A}$ knows the anycast's message anyways (it was unambiguously chosen by $\mathcal{A}$).
    $\Pi^\prime$ does not achieve $MC$, as the published message allows $\mathcal{A}$ to trivially distinguish between scenarios (which only differ in the message to be anycast).
    
    \textbf{P2.}
    Let $\Pi$ be an anycast protocol that achieve $MC$. 
    Let $\Pi^\prime$ be a protocol that behaves identical to $\Pi$, but always selects the first $k$ users in $U_p$ as actual receivers.
    $\Pi^\prime$ still achieves $MC$:
    If one of the first $k$ users is corrupted, $\mathcal{A}$ must have sent the same message in both scenarios an cannot have an advantage in distinguishing.
    Otherwise, it makes no difference to $\mathcal{A}$ which users are selected as actual receivers.
    $\Pi^\prime$ does not achieve $F$, as $\mathcal{A}$ can return the first user of $U_p$ as her guess and always be right.
    
    \textbf{P3, P4.}
    Analogously to P2.
    
    \textbf{P5.}
    Let $\Pi$ be an anycast protocol that achieves $F$. 
    Let $\Pi^\prime$ be a protocol that behaves identical to $\Pi$, publishes $U_a$ after selection.
    $\Pi^\prime$ still achieves $F$, as the adversary receives no protocol output (including the publication of $U_a$) anyways.
    $\Pi^\prime$ does not achieve $WRA$, as now $\mathcal{A}$ does receiver protocol output and can select any user from the published $U_a$ to trivially win.
    
    \textbf{P6.}
    Analogously to P5.
\end{proof}

\section{Protocol}
\label{sec:prot}
Next, we propose \prot, a possible instantiation of anonymous anycast, to demonstrate that our defined notions of privacy are readily achievable.
\prot relies on a unicast channel that unlinks senders from their messages (i.e., the receiver or any outside observer does not learn which message was sent by whom).
This channel is used to provide the anycast sender with randomness, which is used to determine the actual receiver.

\paragraph{Prerequisites}
To send an anonymous anycast using \prot, the following prerequisites have to be met:
\begin{enumerate}
    \item [\textbf{P1.}]
        \prot requires an authenticated {\new and confidential} bidirectional unicast communication channel between the anycast sender and each possible receiver.
        We denote $s$ sending a message $m$ to $r$ over this channel as $\cauth{s,m,r}$.
        Candidates for $\cauth{}$ include the popular Signal messaging application\footnote{\url{https://signal.org} -- Accessed \today}, or email with S/MIME~\cite{RefSMIME}.
    \item [\textbf{P2.}]
        \prot requires a unidirectional unicast communication channel that achieves Sender-Message Pair Unlinkability $(SM)\overline{L}$ as well as confidentiality from every possible receiver to the unicast sender.
        We denote $s$ sending message $m$ to $r$ over this channel as $\canon{s,m,r}$.
        Candidates for $\canon{}$ include Nym (see \Cref{sec:bg:nym}).
\end{enumerate}

\subsection{Basic \prot}
The \prot protocol works in three distinct phases:
In phase $P_0$ (Init), the sender uses $\cauth{}$ to send a \textsc{KeyReq} message to all possible receivers, notifying them of the pending anycast.
The \textsc{KeyReq} message contains instructions for the receivers to contact the sender via $\canon{}$. 

During the second phase $P_1$ (Key Submit), each possible receiver $u$ generates a random symmetric key $k_u$ and sends it to the sender using $\canon{}$.
We cannot expect all receivers to send their keys at exactly the same time, especially if the adversary has the ability to selectively delay packets.
Therefore, we assume that $\canon{}$ has the ability to compensate for delays up to some threshold $T$\footnote{
    The exact value for $T$ depends on the protocol used to initialize $\canon{}$.
}. 
If the keys are delayed within this threshold, $\canon{}$ will still unlink them from their sender.
If the threshold is exceeded and the anycast sender has not received all the keys, it terminates the protocol run.

The third phase $P_2$ (distribution) begins once the sender has received a key from each possible receiver.
{\new First, the sender verifies that all keys are unique.
Then,} the sender chooses a random $k$ from the received keys and uses it to encrypt the message $m$ to be anycast along with a publicly known \textit{tag} which is used to check for correctness after decryption.
The resulting ciphertext is then distributed to all possible receivers using $\cauth{}$.
Each receiver decrypts the ciphertext with their $k_u$ and checks if the revealed \textit{tag} matches the correct one.
If so, the receiver knows that she has been selected as the actual receiver and saves the message.
The \textit{tag} does not serve any privacy or security purpose, it is only used to determine if the message was correctly decrypted {\new in cases where it is not obvious from the revealed plaintext}.
So it can be a fixed byte that is hard-coded into the protocol.
\Cref{fig:atu} visualizes a simplified run of \prot.

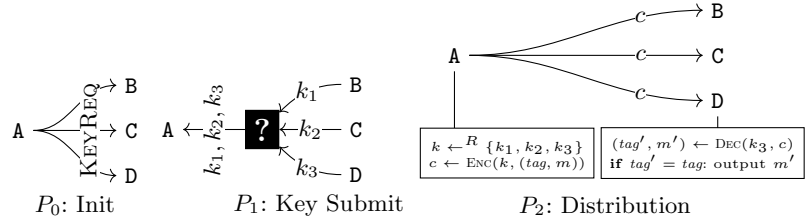
\begin{figure}[t]
    \centering
        \begin{tikzpicture}
        \node[] (s) at (0,0) {\texttt{A}};
        \node[] (r0) at (1.5,0.6) {\texttt{B}};
        \node[] (r1) at (1.5,0) {\texttt{C}};
        \node[] (r2) at (1.5,-0.6) {\texttt{D}};

        \draw[->] (s) to[out=0,in=180] (r0);
        \draw[->] (s) to[out=0,in=180] (r2);
        \draw[->] (s) to[out=0,in=180] node[midway,fill=white,inner sep=0,xshift=0.2cm,font=\small,rotate=90] {\textsc{KeyReq}} (r1);

        \node[] (l) at (0.75,-1) {$P_0$: Init};
    \end{tikzpicture}
    \begin{tikzpicture}
        \node[] (s) at (0,0) {\texttt{A}};
        \node[] (r0) at (2.5,0.6) {\texttt{B}};
        \node[] (r1) at (2.5,0) {\texttt{C}};
        \node[] (r2) at (2.5,-0.6) {\texttt{D}};
        
        \node[draw,fill=black,text=white,font=\large] (a) at (1.25,0) {\textbf{?}};
        \draw[->] (r0) to[out=180,in=45] node[midway,fill=white,inner sep=0] {$k_1$} (a);
        \draw[->] (r1) to[out=180,in=0] node[midway,fill=white,inner sep=0] {$k_2$} (a);
        \draw[->] (r2) to[out=180,in=-45] node[midway,fill=white,inner sep=0] {$k_3$} (a);
        \draw[->] (a) to[out=180,in=0] node[midway,fill=white,inner sep=0,rotate=90] {$k_1,k_2,k_3$} (s);

        \node[] (l) at (2,-1) {$P_1$: Key Submit};
    \end{tikzpicture}
    \begin{tikzpicture}
        \node[] (s) at (0,0) {\texttt{A}};
        \node[] (r0) at (3.5,0.6) {\texttt{B}};
        \node[] (r1) at (3.5,0) {\texttt{C}};
        \node[] (r2) at (3.5,-0.6) {\texttt{D}};

        \draw[] (s) --++ (0,-1.0);
        \node[draw,fill=white,xshift=0.7cm,align=center,font=\tiny] (sl) at (0,-1.3) {
            $k\gets^R \{k_1,k_2,k_3\}$\\
            $c\gets \textsc{Enc}(k,(\textit{tag}, m))$
        };
        \draw[->] (s) to[out=0,in=180] node[near end,fill=white,inner sep=0] {$c$} (r0);
        \draw[->] (s) to[out=0,in=180] node[near end,fill=white,inner sep=0] {$c$} (r1);
        \draw[->] (s) to[out=0,in=180] node[near end,fill=white,inner sep=0] {$c$} (r2);
        \draw[] (r2) --++ (0,-0.5);
        \node[draw,fill=white,xshift=-1.2cm,align=center,font=\tiny] (rl) at (4.5,-1.3) {
            $(\textit{tag}^\prime, m^\prime) \gets \textsc{Dec}(k_3,c)$\\
            \textbf{if} $\textit{tag}^\prime = \textit{tag}$: output $m^\prime$
        };

        \node[] (l) at (2,-2) {$P_2$: Distribution};
    \end{tikzpicture}
    \caption{
        Simplified \prot protocol run.
        \colorbox{black}{\color{white}?} unlinks senders from their messages.
    }\label{fig:atu}
\end{figure}

If the sender wants to anycast to more than one possible receiver, $P_2$ can be repeated $n$ times, discarding previously selected keys.
To send subsequent messages to the same actual receiver, the sender can use the same encryption key as for the initial message and multicast the resulting ciphertext again.

\subsection{Defending against Active Adversaries}
\label{sec:prot:active}
The basic \prot protocol described in the previous section protects against passive adversaries:
$\cauth{}$ ensures that the adversary cannot link keys to receivers, while the encrypted broadcast in phase $P_2$ ensures that the adversary cannot identify the actual receiver from the message sent to her.
Basic \prot, however, cannot achieve confidentiality against \emph{active} attacks.
Consider the following attack:
\begin{enumerate}
    \item During phase $P_1$, $\mathcal{A}$ discards some fraction or even all of the temporary keys submitted to the sender and replaces them with self-chosen ones.
    \item During phase $P_2$, $\mathcal{A}$ intercepts the ciphertext.
        If the sender (unknowingly) chose one of $\mathcal{A}$'s keys, then she can decrypt it and break confidentiality.
\end{enumerate}

If possible receivers add a digital signature to their temporary keys, then $\mathcal{A}$ is no longer able to exchange them for their own keys without the sender noticing.
In our setting, the signature should only reveal that the communication partner is part of the set of possible receivers, not her concrete identity.
To do this, we can use a linkable ring signature scheme~\cite{RefRingSig2,RefRingSig3,RefRingSig4}.
See \Cref{sec:bg:ring_signature} for background on linkable ring signatures. 

To protect against \emph{external} active adversaries, phases $P_0$ and $P_1$ are updated as follows:
In $P_0$, the anycast sender runs \textsc{Sig.Setup} and distributes the public parameters $pp$ as part of the \textsc{KeyReq} message.
After receiving $pp$, each possible receiver $u$ executes \textsc{Sig.KeyGen} to generate its own signing key pair $(vk_u,sk_u)$.
$sk_u$ is stored for future use and $vk_u$ is sent to the anycast sender using $\cauth{}$.
After receiving a verification key from each possible receiver, the anycast sender assembles $R\gets \{vk_1,\dots,vk_l\}$ and sends it to each possible receiver using $\cauth{}$.
Each receiver checks if $R$ contains its verification key and, if so, saves $R$ for later use.
If it does not, the sender is assumed to be malicious and the receiver is dropped from the anycast.

In $P_1$, each possible receiver $u$ generates a signature $\sigma_u$ for its temporary key $k_u$ by executing \textsc{Sig.Sign}$(sk_u,k_u,R)$. 
The receiver $u$ then sends $(k_u,\sigma_u)$ to the anycast sender using $\canon{}$.
Finally, the anycast sender executes \textsc{Sig.Verify}$(\sigma_u,k_u,R)$ on each received key to ensure that all keys were generated by someone within the set of possible receivers (and not an external adversary).

While the steps described above prevent an external adversary from inserting keys, a corrupted possible receiver could expose their private signature key to the adversary. 
Using this key, the adversary can still generate (and validly sign) multiple keys on behalf of the malicious receiver without the sender noticing (assuming the same number of keys from other receivers are dropped by the adversary).
To ensure that each possible receiver can only submit one key, we can use the \emph{linkability} property of the signature:
The anycast sender executes \textsc{Sig.Link}$(\sigma_u, \sigma_{u^\prime})$ for each pair of signatures received.
If \textsc{Sig.Link} returns `$1$' for at least one pair of signatures, the sender detects that \emph{some} malicious possible receiver has sent multiple keys to increase their chance of being selected.
In response, the sender terminates the protocol run.
If \textsc{Sig.Link} returns `$0$' in all cases, the sender proceeds as described above.

Refer to \Cref{alg:prot} for a pseudocode description of \prot.

\begin{algorithm}
\begin{algorithmic}
    \Procedure{Send}{$s,m,n,U_p$}
        \State $pp \gets \textsc{Sig.Setup}(1^\lambda)$
        \State $\textsc{KeyReq} \gets (\textit{hello}, pp)$
        \For{$u\in U_p$}
            \State $\cauth{s,\textsc{KeyReq},u}$ 
        \EndFor
        \State $R\gets \{\}$
        \While{$|R| < |U_p|$}
            \On{$vk$}{$u$}
                \State $R \gets R \cup \{vk\}$
            \EndOn
        \EndWhile
        \For{$u\in U_p$}
            \State $\cauth{s,R,u}$
        \EndFor
        \State $t \gets$ \textsc{Timer}.start()
        \State $K \gets \{\}$
        \State $\Sigma \gets \{\}$
        \While{$|K| < |U_p|$}
            \If{$t \geq T$}
                \State \Return
            \EndIf
            \On{$(k,\sigma)$}{$\bot$}\Comment{$\canon{}$ does not disclose sender.}
                \If{$\textsc{Sig.Verify}(\sigma,k,R) \neq 1$ {\new\textbf{or} $k \in K$}}
                    \State \Return
                \EndIf
                \State $K \gets K \cup \{k\}$
                \State $\Sigma \gets \Sigma \cup \{\sigma\}$
            \EndOn
        \EndWhile
        \For{$\sigma\in \Sigma$}
            \For{$\sigma^\prime \in \Sigma$}
                \If{$\sigma \neq \sigma^\prime \land \textsc{Sig.Link}(\sigma,\sigma^\prime) \neq 0$}
                    \State \Return 
                \EndIf
            \EndFor 
        \EndFor
        \For{$i \in \{1,\dots,n\}$}
            \State $k^\ast \gets^R K$
            \State $K \gets K\setminus \{k^\ast\}$
            \State $c \gets \textsc{Cipher.Enc}((\textit{tag},m),k^\ast)$
            \For{$u\in U_p$}
                \State $\cauth{s,c,u}$
            \EndFor
        \EndFor
    \EndProcedure
    \Procedure{Receive}{$u$}
        \On{\textsc{KeyReq}}{$s$}
            \State $(sk,vk) \gets \textsc{Sig.KeyGen}(pp)$
            \State $\cauth{u,vk,s}$
        \EndOn
        \On{$R$}{$s$}
            \If{$vk \not\in R$}
                \State \Return
            \EndIf
            \State $k \gets \textsc{Cipher.KeyGen}(1^\lambda)$
            \State $\sigma \gets \textsc{Sig.Sign}(sk,k,R)$
            \State $\canon{u,(k,\sigma),s}$
        \EndOn
        \On{$c$}{$s$}
            \State $(\textit{tag}^\prime, m^\prime) \gets \textsc{Cipher.Dec}(c,k)$
            \If{$\textit{tag}^\prime = \textit{tag}$}
                \State \Return $m^\prime$
            \EndIf
        \EndOn
    \EndProcedure
\end{algorithmic}
\caption{
    Sender and receiver behavior for \prot.
    $\textsc{Send}(s,m,n,U_p)$ is executed by user $s$ who wants to send message $m$ to $n$ users out of the set of possible receivers $U_p$.
    $\textsc{Receive}(u)$ is executed by receiver $u$ to receive possible anycast messages. 
    $\lambda$ and \textit{tag} are fixed protocol parameters known to all users.
    During the execution of \textsc{Send}, three lists are assembled:
    $R$ contains the possible receivers' signature public keys, $K$ contains the possible receivers' ephemeral keys, and $\Sigma$ contains their signatures.
}\label{alg:prot}
\end{algorithm}

{\new%
\begin{remark}{(Receiver Impersonation)}
    We do not limit validity of the submitted keys to support (very) asynchronous communication.
    This comes with some security drawbacks:
    For example, an actively malicious possible receiver who was previously a sender within this group of receivers could replace all submitted keys with known keys from the previous protocol run to ensure that she can decrypt the anycast.
    As we assume that senders are honest-but-curious, such attacks are out of scope in this work.
    However, to handle active attacks, one can add a signed timestamp to each submitted key and let the sender discard received keys with too out-of-date timestamps.
\end{remark}
}

\subsection{Security Analysis}
\label{sec:prot:proof}
Finally, we want to show that \prot achieves our privacy notions of message confidentiality, receiver anonymity, and fairness.

We start by proving that \prot achieves message confidentiality.
This is intuitively the case, as messages are encrypted such that only the actual receiver can unveil the plaintext.

\begin{theorem}
    \label{thm:atu:mc}
    \prot achieves message confidentiality against the adeversary $\mathcal{A}$.
\end{theorem}
\begin{proofsketch}
    We need to prove that no efficient adversary $\mathcal{A}$ can win game $\mathcal{G}_{MC}$ with a non-negligible advantage over random guessing.
    We only need to consider the case where $\mathcal{A}$ has not corrupted any possible receiver (nor the sender), as otherwise $\mathcal{A}$ does not receive any protocol output from $\mathcal{C}$ and consequently cannot have an advantage over random guessing.
    
    Assume there exists $\mathcal{A}$ who can break message confidentiality.
    $\mathcal{A}$ has two possible avenues to find out information about the message:
    \begin{enumerate}
        \item Attempt to submit her own keys to the sender in phase $P_1$.
            If one of the adversary's keys is selected, she can regularly decrypt the anycast message in phase $P_2$ and trivially win the game.
        \item Attempt to gain information about the message from observing the distribution in phase $P_2$.
    \end{enumerate}
    Recall that the keys in phase $P_1$ are signed using a linkable ring signature scheme.
    If $\mathcal{A}$ were able to submit a key with a valid signature, it could be used to break the signature scheme's unforgeability property.
    
    Recall that in phase $P_2$ the anycast message in transit is encrypted with keys known only to the sender and actual receivers.
    Thus, if $\mathcal{A}$ can still distinguish the messages based on the observed ciphertexts, it can be used to break the assumed IND-CPA-security of the underlying encryption scheme.
\end{proofsketch}

\begin{theorem}
    \label{thm:atu:wra}
    \prot achieves receiver anonymity against the adversary $\mathcal{A}$.
\end{theorem}
\begin{proofsketch}
    Our goal is to show that there is no efficient $\mathcal{A}$ who has an advantage over random guessing in winning the $\mathcal{G}_{RA}$ game.
    To do so, we iterate through all of $\mathcal{A}$'s abilities as listed in \Cref{sec:prob:adv} and argue that none of them helps her to gain an advantage.
    
    \begin{itemize}
        \item \textit{Passive Observation.}
            Passive observation allows $\mathcal{A}$ to analyze incoming and outgoing packets anywhere in the network.
            During phase $P_0$, the sender sends an identical \textsc{KeyReq} package to every possible receiver.
            As all packages are identical, they cannot contain information about any actual receiver.
            We can analogously argue for phase $P_2$, where the sender sends an identical ciphertext to all possible receivers.
            In phase $P_1$, each possible receiver sends a unique key to the sender.
            If $\mathcal{A}$ were able to track who sends which key, she could identify the actual receiver based on their key and the ciphertext from phase $P_2$.
            However, if this were the case, $\mathcal{A}$ would also be able to break $(SM)\overline{L}$ for the anonymous unicast channel, which contradicts our assumptions.
        \item \textit{Timing.}
            $\mathcal{A}$ can time sending behavior of any user.
            In phases $P_0$ and $P_2$, packets are sent simultaneously by the sender to the receiver via a multicast message.
            In $P_2$, the sender selects a random key from the received ones and encrypts using the selected key prior to sending.
            As all keys randomly chosen and of equal length, we can assume that this selection and encryption do not vary in the time it takes based on which key is selected.
            If $\mathcal{A}$ were able to utilize timing in $P_1$ to identify actual receivers, she could also break $(SM)\overline{L}$ for the anonymous unicast channel.
        \item \textit{Active Interference.}
            $\mathcal{A}$ can actively interfere (i.e., drop, delay, modify, insert, and replay) with arbitrary packets.
            We can rule out that any active interference helps $\mathcal{A}$ in linking keys to receivers in phase $P_1$, as this would break $(SM)\overline{L}$ in the underlying unicast protocol.
            \begin{itemize}
                \item \textit{Drop.}
                    If $\mathcal{A}$ drops \textsc{KeyReq} messages, the receiving clients are not informed about being possible receivers and will not participate further in the protocol.
                    As the sender will only choose actual receivers if she has received the expected number of keys, the anycast will not be executed.
                    The same behavior occurs if $\mathcal{A}$ drops keys in $P_1$.
                    If $\mathcal{A}$ drops ciphertexts in $P_2$, the actual receivers may not receive the message.
                    However, as we assume that receivers show no outward reaction to received (or not received) data, this does not reveal any information about the identity of possible receivers to $\mathcal{A}$.
                \item \textit{Delay.}
                    Delays of packets other than the receivers' keys have no effect other than prolonging the protocol execution, as users wait for all expected packets to arrive before continuing with the execution.
                    Delays of receivers' keys within the threshold $T$ disclose no additional information as we assume that $\canon{}$ can compensate for these delays.
                    Delays in excess of $T$ cause the anycast sender to terminate the run prior to selecting the actual receiver and thus cannot reveal any information about the actual receiver to $\mathcal{A}$.
                \item \textit{Modify.}
                    Due to the use of MACs, \textsc{KeyReq}s cannot be modified without detection.
                    If keys in $P_1$ or ciphertexts in $P_2$ are modified, actual receivers might not be able to successfully decrypt the message.
                    However, as we assume no external reaction from the receivers, this does not reveal any information to $\mathcal{A}$.
                \item \textit{Insert.}
                    In phase $P_0$, $\mathcal{A}$ cannot insert further valid \textsc{KeyReq}s due to the use of MACs.
                    In phase $P_1$, $\mathcal{A}$ cannot insert further keys, as the sender only proceeds with the anycast if the expected number of keys arrives.
                    If $\mathcal{A}$ drops keys to insert its own ones, it is not able to do so without detection due to the use of linkable ring signatures for the key messages.
                    In phase $P_2$, $\mathcal{A}$ may insert new ciphertexts, but the receiving clients will not show any outward reaction.
                \item \textit{Replay.}
                    Replaying \textsc{KeyReq}s in phase $P_0$ will only lead to the receiving clients discarding any extra ones.
                    Replaying keys in $P_1$ will result in the sender receiving more keys than expected and not executing the anycast as a result.
                    Replaying ciphertexts in $P_2$ elicits no reaction from the receivers by assumption.
            \end{itemize}
        \item \textit{User Corruption.}
            $\mathcal{A}$ can corrupt the sender as well as a fraction of possible receivers.
            Recall that $\mathcal{A}$ only receives protocol output (and therefore has a chance to not randomly guess) in $\mathcal{G}_{RA}$ if no actual receiver is corrupted and there exists at least one other possible receiver who is not corrupted.
            
            By corrupting a receiver, $\mathcal{A}$ gains access to their internal state, including all key material.
            However, as the anycast message is only encrypted with the keys of actual receivers (who are not corrupted) and there remain honest possible receivers, $\mathcal{A}$ cannot use the gained information to determine which honest clients are actual receivers versus non-chosen possible receivers. 

            $\canon{}$'s achievement of Sender-Message Pair unlinkability ensures that a corrupted anycast sender cannot link received keys to their owner.
            The ring signature scheme's signer anonymity property ensures that the corrupted anycast sender cannot link based on the key's ring signature.
    \end{itemize}
    Thus, we have argued that none of $\mathcal{A}$'s abilities help her in winning $\mathcal{G}_{RA}$.
\end{proofsketch}

As fairness is implied by receiver anonymity (\Cref{thm:imp1}), \Cref{thm:atu:wra} also implies that \prot achieves fairness.

\section{Evaluation}
\label{sec:eval}
We evaluate the performance of \prot concerning two metrics:
\begin{enumerate}
    \item Computational overhead for senders and receivers (\Cref{sec:eval:mb}).
    \item End-to-end latency between sender and actual receiver (\Cref{sec:eval:lat}). 
\end{enumerate}
We suspect that the end-to-end latency of \prot largely depends on that of the underlying anonymous channel and hence benefits from improvements in this field that is in dynamic development right now.
To get a more robust view of latency, we are interested in long-term latency data on Nym, one possible instantiation of \prot's anonymous channel.
We analyze this data in \Cref{sec:eval:nym}.

\subsection{Computational Overhead}
\label{sec:eval:mb}
We want to determine how the computational overhead of \prot for sender and receivers scales with the number of possible receivers.
A \prot protocol run can be divided into several distinct steps.
Each step consists of different cryptographic operations and might scale differently with the number of receivers.
Thus, we construct a series of microbenchmarks with which we can measure each step separately.

We have split protocol execution between sender and receivers into six distinct steps to gain insight into the contribution to the overhead of different components:
\begin{itemize}
    \item \textit{Key Generation.}
    Step KG is executed by each receiver in phase $P_0$ and entails the generation of one linkable ring signature key pair.
    \item\textit{Sign.}
    Step SIG is executed by each receiver in phase $P_1$ and entails the generation of a 32-byte AES key and the signing of it.
    \item\textit{Verfiy.}
    Step VER is executed by the anycast sender in phase $P_1$.
    During this step, the sender verifies each received ring signature.
    \item\textit{Link.}
    Step LINK is executed by the anycast sender in phase $P_1$.
    During this step, the sender tests for each pair of signatures if they are linked.
    \item\textit{Select \& Encrypt.}
    Step S\&E is executed by the anycast sender in phase $P_2$ and entails the selection of receiver key(s) as well as the encryption of the message.
    \item\textit{Decrypt \& Compare.}
    Step D\&C is executed by each receiver in phase $P_2$ and entails the decryption of the received ciphertext as well as the comparison of the revealed tag with a fixed value.
\end{itemize}
We have implemented a prototype in go, which can be found on GitHub\footnote{
    \url{https://github.com/coijanovic/anycast-bench}
}.
For all our measurements, we use a virtual machine running Ubuntu 22.04.1 on a server with an AMD EPYC 7502 Processor, 2 assigned cores, and 4GB of RAM.
We use \texttt{lirisi}\footnote{\url{https://github.com/zbohm/lirisi} -- Accessed \today} as our linkable ring signature, which implements a signature scheme proposed by Liu et al.~\cite{RefRSImp}.
For symmetric encryption, go's standard \texttt{crypto/aes} package is used.
To determine the impact of the number of possible receivers on the computational overhead, we execute each step for 10, 20, and 40 possible receivers and one actual receiver.
In all experiments, a 1024-byte message of random content is used.

\paragraph{KG}
We expect the computational overhead of step KG to be independent of the number of possible receivers:
Each receiver has to generate one standard elliptical curve key pair, regardless of the total number of receivers.
\paragraph{SIG}
We expect computational overhead for step SIG to linearly scale with the number of possible receivers, as the signature depends on each element of the ring.
In general, we expect the signature generation to require the greatest share of computational time for the receiver.
\paragraph{VER}
We expect the computational overhead for step VER to quadratically scale with the number of possible receivers:
\begin{itemize}
    \item The number of signatures to verify grows linearly with the number of receivers
    \item Verification of \texttt{lirisi}'s signatures requires computations for each of the ring's public keys.
    The number of public keys in the ring grows linearly with the number of receivers.
\end{itemize}
In general, we expect the signature verification to require the greatest share of computational time for the sender.
\paragraph{LINK}
We expect the computational overhead for step LINK to scale quadratically with the number of possible receivers, as every signature has to be compared to every other one.
\paragraph{S\&E}
We expect computational overhead for step S\&E to scale linearly with the number of possible receivers, as the set of keys to select from grows with the number of receivers.
\paragraph{D\&C}
Finally, we expect computational overhead for step D\&C to be independent of the number of possible receivers, as the decryption times of block ciphers (such as AES) only depend on the size of the ciphertext~\cite{RefAESPerf2} and the key size~\cite{RefAESPerf1}.
Both parameters are identical in all experiments.

We provide mean values as well as standard deviation of 100 individual measurements for each step in \Cref{tab:microbench}.
One can see that the ring signature verification has by far the biggest impact on the computational overhead for the sender.
For receivers, the largest contributor to computational overhead is also the linkable ring signature scheme, which is used to sign the keys.

For steps KG, SIG, VER, and D\&C the measurements confirm our expectations.
Our measurements for step LINK show superlinear, but no clear quadratic growth:
For 10 receivers, linking takes $0.8{\mu}s$ and the time increases by a factor of about $3.3$ for every doubling of the receiver.
In the \texttt{lirisi} signature scheme, each linking requires equality checks of two signature components.
We suspect that compiler optimizations are to blame for the observed discrepancy.
Step S\&E takes about $1.8{\mu}s$ independent of the number of possible receivers, which also does not match our expectations.
Note the comparatively high standard deviation in our measurements of this step.
Due to the low execution time, CPU scheduling can have a large impact on the measurement.
We thus do not recommend relying on these measurements as absolute values but rather as a comparison to the other steps.
We suspect that the linear growth due to the key selection process is hidden by this deviation and is revealed only with (much) larger numbers of receivers.
The same caveat about measurement deviation also applies in step D\&E.

As we suspected, signature generation and verification are responsible for the largest part of computational overhead by far.
Verification further scales quadratically in the number of receivers, which limits scalability.
On the bright side, we have seen that \prot's other computational steps are very lightweight.
Future---more efficient---linkable ring signature schemes, therefore, have the potential to also make \prot equally more efficient.

Note that we have evaluated anycast to a \emph{single} actual receiver.
For $n$ actual receivers, the computational overhead for steps S\&E and D\&C increases by a factor of $n$, as it has to be repeated for every actual receiver.
The overhead for the other steps is independent of the number of actual receivers.

\begin{table}[]
    \centering
    \begin{tabularx}{\columnwidth}{rXXX}\toprule
    & \multicolumn{3}{c}{\textbf{\# Possible Receivers}}\\
    & 10 & 20 & 40\\\midrule
    \textbf{KG} & $15.79 {\mu}s \pm 0.95$ & $15.96 {\mu}s \pm 1.02$ & $15.88 {\mu}s \pm 1.17$\\
    \textbf{SIG} & $2.70 ms \pm 0.12$ & $5.45 ms \pm 0.19$ & $10.85 ms \pm 0.22$\\
    \textbf{VER} & $28.07 ms \pm 1.12$ & $109.78 ms \pm 2.84$ & $432.18 ms \pm 5.76$\\
    \textbf{LINK} & $0.80 {\mu}s \pm 0.175$ & $2.67 {\mu}s \pm 0.45$ & $8.62 {\mu}s \pm 0.95$\\
    \textbf{S\&E} & $1.82 {\mu}s \pm 0.59$ & $1.69 {\mu}s \pm 0.39$ & $1.63 {\mu}s \pm 0.38$\\
    \textbf{D\&C} & $0.85 {\mu}s \pm 0.27$ & $0.75 {\mu}s \pm 0.14$ & $0.74 {\mu}s \pm 0.12$\\\bottomrule
    \end{tabularx}
    \caption{
        Results of \prot microbenchmarks.
        All results are mean values of 100 runs $\pm$ standard deviation.
    }
    \label{tab:microbench}
\end{table}

\subsection{Nym Latency Measurements}
\label{sec:eval:nym}
An obvious instantiation of \prot's anonymous channel is Nym, as it is a recent design and provides a free public instance.
There are two hypotheses we want to test to determine Nym's suitability as an anonymous channel for \prot:
\begin{enumerate}
    \item Due to its mixnet architecture, we suspect that Nym's latency is in general much higher that that of non-anonymous channels.
        If Nym's latency accounts for a large share of \prot's end-to-end latency, \prot's usability directly depends on Nym.
    \item Due to the large number of possible paths a message can take through Nym's mix network, we suspect that Nym's latency is subject to high variance.
\end{enumerate}
We test these hypotheses by conducting a long-term measurement study of Nym's latency.
Over 32 days, we sent a total of 3745 32-byte messages of random content between two Nym WebSocket clients\furl{https://nymtech.net/docs/stable/integrations/websocket-client} with the default configuration and of the latest version (v1.0.2 for early measurements and v1.1.0 for later ones).
For each message, we logged the time at which it was passed to the sending client and the time when it was output by the receiving client.
Our dockerized experiment setup can be found on GitHub\furl{https://github.com/coijanovic/nym-latency-observer}.

During the measurement period, our server experienced a short network outage, which resulted in very high latency for a small number of messages.
As this was a local issue and does not reflect the latency caused by Nym itself, we discard all outliers in our data with latency over ten seconds.
We further want to highlight a 5-day gap in our measurements (see \Cref{fig:days}).
During this time, the Nym platform was updated and communication between our clients was not possible.

We observed a median end-to-end latency of 448.51 ms with a 95th percentile latency of 861.46 ms and a 99th percentile latency of 3026.03 ms. 
The minimum observed latency was 99.28 ms, while the maximum was 9594.46 ms.
The standard deviation was measured at 667.52.
\Cref{fig:weekdays,fig:hours,fig:days} present latency measurements broken down by day of the week, day of the month, and hour of the day respectively.

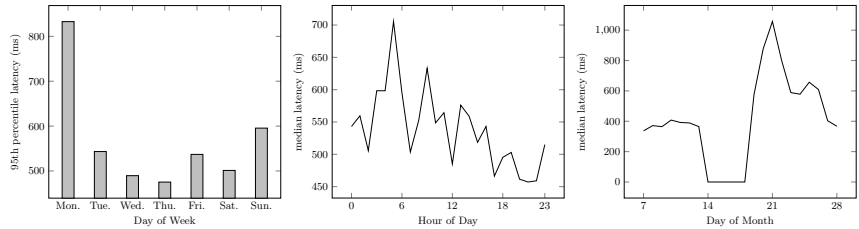
\begin{figure}
    \centering
    \begin{subfigure}[t]{0.30\textwidth}
        \begin{tikzpicture}[thick, scale=0.45]
\begin{axis}[
        symbolic x coords={},
        symbolic x coords={Mon., Tue., Wed., Thu., Fri., Sat., Sun.},
        xtick=data,
        xlabel={Day of Week},
        ylabel={95th percentile latency (ms)},
        legend pos=north west,
        legend style={draw=none}]
    ]
    \addplot[ybar,black,fill=black!25] coordinates {
        (Mon.,832.8660789858491)
        (Tue.,543.259425733777)
        (Wed.,489.47196513798707)
        (Thu.,475.15422054166663)
        (Fri.,536.912603835034)
        (Sat.,501.18264904863224)
        (Sun.,595.4509077763819)
    };
\end{axis}
\end{tikzpicture}
        \caption{
            Per weekday.
        }
        \label{fig:weekdays}
    \end{subfigure}
    \begin{subfigure}[t]{0.30\textwidth}
        \begin{tikzpicture}[thick, scale=0.45]
\begin{axis}[
        xtick={0,6,12,18,23},
        xlabel={Hour of Day},
        ylabel={median latency (ms)},
        legend pos=north west,
        legend style={draw=none}]
    ]
    \addplot[black] coordinates {
        (0,543.0583969226191)
        (1,559.5866057068965)
        (2,505.71219826829264)
        (3,598.372662802721)
        (4,598.3382197109374)
        (5,705.4578329153846)
        (6,595.1884215149254)
        (7,503.55648331578936)
        (8,552.7518211449275)
        (9,632.8579458958334)
        (10,548.6546284718308)
        (11,564.3662840551724)
        (12,485.2244542307692)
        (13,576.0300260812501)
        (14,558.786564)
        (15,518.6122875176471)
        (16,542.9754268705883)
        (17,466.2631563888889)
        (18,495.41022964071857)
        (19,502.88828799397595)
        (20,461.5306787939394)
        (21,457.19071851807234)
        (22,459.04455454437874)
        (23,515.073644006024)
    };
\end{axis}
\end{tikzpicture}
        \caption{
            Per hour of the day.
        }
        \label{fig:hours}
    \end{subfigure}
    \begin{subfigure}[t]{0.30\textwidth}
        \begin{tikzpicture}[thick, scale=0.45]
\begin{axis}[
        xtick={7,14,21,28},
        xlabel={Day of Month},
        ylabel={median latency (ms)},
        legend pos=north west,
        legend style={draw=none}]
    ]
    \addplot[black] coordinates {
        (7,337.0275501969696)
        (8,371.36731684027774)
        (9,364.43702006944443)
        (10,407.81198759245285)
        (11,392.05525155208335)
        (12,388.86096732638885)
        (13,364.83036506849317)
        (14,0)
        (15,0)
        (16,0)
        (17,0)
        (18,0)
        (19,576.0120548223684)
        (20,877.6067672905027)
        (21,1057.0122632328769)
        (22,800.5467490044446)
        (23,588.8199870104529)
        (24,578.2249300780669)
        (25,657.1532194156978)
        (26,609.2583962316603)
        (27,403.55763261594205)
        (28,367.1386413409091)
    };
\end{axis}
\end{tikzpicture}
        \caption{
            Per day of the month.
        }
        \label{fig:days}
    \end{subfigure}
    \caption{Median end-to-end latency in Nym.}
\end{figure}

Our collected data corroborates the hypothesis that Nym's latency is much higher than that of non-anonymous communication:
While standard internet latency is commonly around 20 to 30 ms\furl{https://www.statista.com/statistics/1244676/}, the use of Nym increased latency more than tenfold to a median of 0.45 s.

We can also confirm the hypothesis that Nym's latency is highly variable:
The median latency we observed on Mondays was nearly twice as high as the latency on Thursdays.
The minimum and maximum observed latency also differed by a factor of 100.
We can explain Nym's high latency variance by its architecture:
Clients' messages pass through multiple mix servers before arriving at the sender.
If a path is chosen where network latency is high between servers and servers are under high load, latency will naturally be higher than with close and idle servers.

We have seen that Nym's latency varies significantly over time.
We do not claim to be able to show any long-term trends in Nym's latency (e.g., `latency is higher on Mondays'), as we did not collect data over a long enough period for that.
With sub-second latency for 95\% of cases, we still conclude that Nym is suited for use in \prot, at least in settings without real-time communication requirements.

\subsection{End-To-End Latency}
\label{sec:eval:lat}
In \Cref{sec:intro}, we suggested that anonymous anycast can be used by political activists to implement a dead man's switch.
We expect the dead man's notification to be similar in size and expected latency to instant messaging.
ITU Recommendation G.1010\furl{https://www.itu.int/rec/T-REC-G.1010-200111-I} states that ``delays of several seconds are acceptable'' for instant messaging applications.
We thus want to determine how latency in \prot scales with the number of receivers and size of the message and if it falls within the ITU's recommended latency limits.

We define the end-to-end latency as the time difference between the start of the sending client and the plaintext output of the actual receiver.
To determine the impact of the number of receivers on the end-to-end latency, we will run an experiment with a fixed message size of 512 Byte and vary the number of receivers between 4 and 16. 
To determine the impact of the message size on the end-to-end latency, we will run an experiment with a fixed number of receivers of 8 and vary the message size between 512 Byte and 2 KB.
To enable these experiments, we implemented a prototype of \prot in go\furl{https://github.com/coijanovic/panini}.
The authenticated channel was instantiated with TCP connections over which ecdsa-signed messages were sent.
The anonymous channel was instantiated with Nym (WebSocket client version v1.1.1).
Each client ran in a separate docker container with its own Nym client.
One of the clients acted as the sender and the rest as possible receivers (of which \emph{one} was chosen as the actual receiver).
All containers ran on an AMD Ryzen 5 5600G with 32 GB of RAM.
We repeat each measurement 16 times and present the median of the observed latencies $\pm$ standard deviation.

We expect \prot's latency to be largely independent of both the number of receivers as well as the message size.
As we have shown in \Cref{sec:eval:mb}, computational times for senders and receivers are well below 0.5s.
We thus expect \prot's end-to-end latency to be dominated by the latency of Nym.
Message size only impacts phase $P_2$, where the message is encrypted, sent via the authenticated channel, and decrypted by the receivers.
We have already shown that en- and decryption requires computation on the order of microseconds, sending over the authenticated channel (Nym in this case) also should only add minimal latency.
The more receivers participate, the more connection over the anonymous channel have to be made.
While receivers can send their data in parallel, the sender has to wait for the slowest receiver before continuing the execution.
If there is variance in the latency of the anonymous channel, we can expect \prot's end-to-end latency to increase with the number of receivers.

We measured a median end-to-end latency of $0.71 s \pm 0.64$ for 4 receivers, $0.76 s \pm 0.34$ for 8 receivers, and $0.82 s \pm 2.13$ for 16 receivers.
For the message size experiments, we measured $0.65 s \pm 0.62$ end-to-end latency for 256 B message, $0.76 s \pm 0.34$ for 512 B, $0.84 s \pm 0.83$ for 1024 B, and $0.75 s \pm 0.54$ for 2048 B.

As we expected, latency increases with the number of possible receives.
For the message size experiments, we suspect that Nym's latency variance is to blame for the unexpected results:
As sending 2048 B messages leads to lower latency than sending 1024 B messages, it seems unlikely that the message size itself is to blame.
As we have already seen in \Cref{sec:eval:nym}, Nym's latency fluctuates over time, depending e.g., on network utilization.
If the measurements for 2048 B messages were made during a period of lower utilization than the measurements for 1024 B messages, our results can be explained.
Finally, the high standard deviation we observed in our measurements can also be explained by Nym's latency variation.

In summary, we have shown that---for up to 16 receivers and 2 KB messages---\prot achieves sub-second end-to-end latency and is therefore suitable for instant messaging applications according to the ITU's recommendation.
For our considered message sizes, we have determined that the anonymous channel is the bottleneck:
Nym's median latency of 0.45 s (see \Cref{sec:eval:nym}) accounts for nearly 60\% of \prot's end-to-end latency for 8 possible receivers.
While we have only evaluated \prot with Nym, we want to note that the two protocols are not inherently linked to each other.
If a future anonymous communication network that achieves sender-message pair unlinkability with lower latency is proposed, \prot can utilize it, lowering its end-to-end latency in turn.

\vspace{0.25cm}
If \prot is used in settings outside of instant messaging where the message size is much larger (e.g., to distribute multi-Gigabyte documents in whistleblowing), the ciphertext distribution in phase $P_2$ might become a bottleneck.
Recall that for 8 receivers and a 512 Byte message, we measured an end-to-end latency of 0.76 s.
Assume that the sender has a 100 Mb/s internet connection.
The time needed to distribute 512 Byte to 8 receivers over this connection is negligible.
Thus we can calculate the total latency $\Delta$ roughly as follows:
\[
    \Delta = 0.76s + \frac{\text{\#Receivers}\times\text{message size}}{\text{transfer rate}}
\]
We can see that for any message larger than 9.5 MB, distributing the ciphertext to 8 receivers over the 100 Mb/s connection requires more time than the remainder of \prot's execution.
However, this issue can be circumvented by pre-distributing an encrypted version of the data and using the anycast only to provide the required decryption key to the actual receiver.

\section{Conclusion}
\label{sec:conc}
In this paper, we have identified and formally defined message confidentiality, receiver anonymity, and fairness as the main privacy goals of anonymous anycast.
Based on our formal definitions, it is now possible to provide rigorous proof of privacy for anonymous anycast protocols.
We have further introduced \prot, the first protocol that enables anonymous anycast over readily available infrastructure.
We have provided proof that \prot fulfills all of our previously defined privacy goals.
In an in-depth empirical evaluation, we have shown that \prot only introduces minimal computational overhead for anycast senders and receivers and achieves end-to-end latency suitable for instant messaging.

\bibliographystyle{splncs04}
\bibliography{ref.bib}


\end{document}